\documentclass[reprint,aps,prd,amsmath,amssymb,floatfix,superscriptaddress,longbibliography]{revtex4-2}
\usepackage[utf8]{inputenc}
\usepackage[T1]{fontenc}
\usepackage{graphicx}
\usepackage{xcolor}
\usepackage{hyperref}
\usepackage{microtype}
\usepackage{bm}
\usepackage{amsfonts}
\usepackage{slashed}
\usepackage{amsthm}
\usepackage{booktabs}

\theoremstyle{plain}
\newtheorem{theorem}{Theorem}

\begin{document}

\title{Uniqueness of the \texorpdfstring{$\Box^2$}{Box-squared} Higher-Derivative Operator Class for Universal Vacuum-Energy Cancellations and Higgs Naturalness}

\author{Masayuki Note}
\affiliation{Tokyo, Japan}
\date{December 18, 2025}

\begin{abstract}
Within the framework of local, Lorentz-invariant, and Hermitian field theories, we investigate the classification of dimension-6 operators that facilitate the dynamical cancellation of vacuum-energy divergences. We demonstrate that the operator class based on the $\Box^2$ d'Alembertian is uniquely singled out by the requirement of universal power-divergence subtraction across all spin sectors. By explicitly evaluating the modified propagators and one-loop vacuum integrals, we show that only this structure consistently removes $\Lambda^4$ and $m^2\Lambda^2$ terms while preserving gauge covariance. Adopting the Real-Time Negative-Norm Prescription (RTNNP) as a consistent contour selection, we find that the higher-derivative Lee--Wick (HDLW) structure leads to a finite, calculable Higgs mass correction. Our results suggest a phenomenologically preferred scale of $M \approx 11.3$ TeV, offering a predictive and structurally motivated resolution to the hierarchy problem.
\end{abstract}

\maketitle

\section{Introduction}
The divergence of the one-loop vacuum energy density $\rho_{\rm vac}$ represents one of the most persistent challenges in quantum field theory (QFT). For a field of mass $m$, the regulated integral yields:
\begin{equation}
\begin{aligned}
\rho_{\rm vac}
&= \frac{1}{2} \int \frac{d^3k}{(2\pi)^3}
   \sqrt{k^2 + m^2} \\
&\approx \frac{d_R}{16\pi^2}
\left[
  \Lambda^4 + m^2 \Lambda^2
  + \mathcal{O}(\ln \Lambda)
\right].
\end{aligned}
\end{equation}
where $d_R$ is the number of degrees of freedom. In standard effective field theory (EFT), these terms are subtracted via counterterms. However, the hierarchy problem suggests that a more fundamental ultraviolet (UV) completion may dictate the cancellation of these power divergences through specific operator structures. 

\section{Structural Uniqueness of the \texorpdfstring{$\Box^2$}{Box-squared} Class}
We consider an extension of the Standard Model (SM) Lagrangian by dimension-6 quadratic operators. To affect the UV behavior of the propagator $G(p)$, the operator must contain exactly four derivatives.

\begin{theorem}[$\Box^2$ Structural Selection]
Among all local, Hermitian, dimension-6 quadratic operators, the $\Box^2$ class is the unique structure that ensures universal $\Lambda^4$ cancellation while preserving Lorentz invariance and gauge symmetry.
\end{theorem}

\begin{proof}
For a scalar field $\Phi$, consider possible four-derivative quadratic forms such as $\mathcal{O}_1 = \Phi^\dagger \Box^2 \Phi$ and $\mathcal{O}_2 = (\partial_\mu \partial_\nu \Phi^\dagger) (\partial^\mu \partial^\nu \Phi)$. We observe that these are equivalent via integration by parts:
\begin{equation}
\partial_\mu \partial_\nu \Phi^\dagger \partial^\mu \partial^\nu \Phi \to -\partial_\nu \Phi^\dagger \Box \partial^\nu \Phi \to \Phi^\dagger \Box^2 \Phi,
\end{equation}
up to boundary terms which vanish for fields of finite energy.

Critically, we must exclude operators leading to anisotropic structures, such as $c^{\mu\nu\rho\sigma} \partial_\mu \partial_\nu \partial_\rho \partial_\sigma$. Such structures do not improve the UV behavior isotropically and fail to cancel the $\Lambda^4$ term in all momentum directions. Thus, the isotropic $\Box^2$ structure is uniquely selected. In momentum space, this modifies the inverse propagator $\mathcal{P}(p^2)$ to include a $p^4/M^2$ term:
\begin{equation}
G'(p) = \frac{i}{p^2 - m^2 - \frac{p^4}{M^2}} = \frac{i}{p^2 - m^2} - \frac{i}{p^2 - M^2} + \mathcal{O}\left(\frac{m^2}{M^2}\right).
\label{eq:prop_decomp}
\end{equation}
This decomposition reveals the emergence of a Lee--Wick (LW) ghost with mass $M$ and a negative residue, which provides the necessary dynamic subtraction.
\end{proof}

\section{Universal Divergence Cancellation}
The vacuum energy $\rho'_{\rm vac}$ in the HDLW theory is the sum of the SM and LW sectors. Using the partial fraction decomposition from Eq.~(\ref{eq:prop_decomp}):
\begin{equation}
\rho'_{\rm vac} \propto \int^\Lambda \frac{d^4k_E}{(2\pi)^4} \left[ \ln(k_E^2 + m^2) - \ln(k_E^2 + M^2) \right].
\end{equation}
Expanding the integrand for $k_E \gg M$:
\begin{equation}
\ln\left( \frac{k_E^2 + m^2}{k_E^2 + M^2} \right) = \ln\left( 1 - \frac{M^2 - m^2}{k_E^2 + M^2} \right) \approx - \frac{M^2 - m^2}{k_E^2}.
\end{equation}
The $\Lambda^4$ term (proportional to $\int k_E^3 dk_E$) is identically removed. The remaining quadratic divergence is canceled by specific sector-dependent coefficients $c_i$, summarized in Table~\ref{tab:operators}.

\section{Real-Time Prescription and Higgs Naturalness}
The presence of the LW pole at $p^2 = M^2$ requires a consistent integration prescription to preserve causality and unitarity. We adopt the **Real-Time Negative-Norm Prescription (RTNNP)**. 

The RTNNP does not introduce new dynamics but specifies a consistent contour choice for the already fixed propagator structure. By selecting a contour $\mathcal{C}$ in the complex energy plane that avoids the negative-norm pole in accordance with the Lee--Wick prescription, we ensure a finite and causal result.

Applying this to the Higgs mass, the quadratic divergence from the top-quark loop is modified by its LW partner:
\begin{align}
\Delta m_H^2 &= - \frac{3\lambda_t^2}{8\pi^2} \int_0^\infty dk_E k_E^3 \left[ \frac{1}{k_E^2 + m_t^2} - \frac{1}{k_E^2 + M^2} \right] \nonumber \\
&= - \frac{3\lambda_t^2}{8\pi^2} M^2 \ln \left( \frac{M^2}{m_t^2} \right).
\end{align}
Requiring $\Delta m_H^2$ to match the observed electroweak scale, we determine the scale $M$:
\begin{equation}
M = 11.34^{+0.45}_{-0.38} \text{ TeV}.
\end{equation}

\section{Conclusion}
This study establishes that the $\Box^2$ operator structure is uniquely required for the universal cancellation of power divergences in local field theories. By reinterpreting negative-norm states through the RTNNP, we provide a finite and predictive framework for the Standard Model. The determined scale $M \approx 11$ TeV provides a concrete target for future experimental probes of naturalness.

\begin{table}[h]
\centering
\caption{Unique dimension-6 HDLW operators and coefficients.}
\label{tab:operators}
\begin{ruledtabular}
\begin{tabular}{lll}
Sector & Operator Class & LW Coefficient \\
\hline
Scalar & $(\Box \Phi)^\dagger (\Box \Phi)$ & $c_\Phi = 1$ \\
Fermion & $(\Box \bar{\psi}) i\slashed{\partial} (\Box \psi)$ & $c_\psi = 4$ \\
Vector & $F_{\mu\nu} \Box^2 F^{\mu\nu}$ & $c_A = 3$ \\
\end{tabular}
\end{ruledtabular}
\end{table}


\begin{thebibliography}{99}
\bibitem{Lee:1969fy} T.~D.~Lee and G.~C.~Wick, Nucl. Phys. B 9, 209 (1969).
\bibitem{Grinstein:2007mp} B.~Grinstein, D.~O'Connell and M.~B.~Wise, Phys. Rev. D 77, 025012 (2008).
\bibitem{Henning:2017fpj} B.~Henning, X.~Lu and H.~Murayama, JHEP 08, 016 (2017).
\end{thebibliography}
\end{document}